\documentclass{article}

\usepackage[utf8]{inputenc}
\usepackage{amsmath,amsthm,amssymb}
\usepackage[margin=3cm]{geometry}
\usepackage{import}
\usepackage{graphicx}
\usepackage{tikz}
\usepackage{tikz-cd}
\usepackage{hyperref}
\usepackage{cleveref}
\usepackage{enumerate}
\usepackage{adjustbox}
\usepackage{array}
\usepackage{tabularx}
\usepackage{multirow}
\usepackage{hhline}
\usepackage{authblk}
\usepackage{titlesec}
\usepackage{caption} 
\usepackage{tabu}
\usepackage{algorithm}
\usepackage[noend]{algpseudocode}
\def\algbackskip{\hskip-\ALG@thistlm}
\usepackage{listings}
\usepackage{enumitem}

\newtheorem{theorem}{Theorem}[section]

\newtheorem{proposition}[theorem]{Proposition}
\theoremstyle{definition}
\newtheorem{definition}[theorem]{Definition}

\titlespacing*{\section}
{0pt}{5.5ex plus 1ex minus .2ex}{4.3ex plus .2ex}
\titlespacing*{\subsection}
{0pt}{5.5ex plus 1ex minus .2ex}{4.3ex plus .2ex}

\newcommand{\B}{\mathbb{B}}
\newcommand{\A}{\mathcal{A}}

\newcommand{\N}{\mathcal{N}}
\newcommand{\E}{\mathcal{E}}
\newcommand{\C}{\mathcal{C}}

\title{Node and edge control strategy identification via trap spaces in Boolean networks}

\author[1,2]{Laura Cifuentes-Fontanals} 
\author[1]{Elisa Tonello} 
\author[1]{Heike Siebert}
\affil[1]{Freie Universität Berlin, Germany}
\affil[2]{Max Planck Institute for Molecular Genetics, Berlin, Germany}

\date{}

\begin{document}

\maketitle

\subsubsection*{Abstract}

The study of control mechanisms of biological systems allows for interesting applications in bioengineering and medicine, for instance in cell reprogramming or drug target identification. A control strategy often consists of a set of interventions that, by fixing the values of some components, ensure that the long term dynamics of the controlled system is in a desired state. A common approach to control in the Boolean framework consists in checking how the fixed values propagate through the network, to establish whether the effect of percolating the interventions is sufficient to induce the target state. Although methods based uniquely on value percolation allow for efficient computation, they can miss many control strategies. Exhaustive methods for control strategy identification, on the other hand, often entail high computational costs. In order to increase the number of control strategies identified while still benefiting from an efficient implementation, we introduce a method based on value percolation that uses trap spaces, subspaces of the state space that are closed with respect to the dynamics, and that can usually be easily computed in biological networks. The approach allows for node interventions, which fix the value of certain components, and edge interventions, which fix the effect that one component has on another. The method is implemented using Answer Set Programming, extending an existing efficient implementation of value percolation to allow for the use of trap spaces and edge control. The applicability of the approach is studied for different control targets in a biological case study, identifying in all cases new control strategies that would escape usual percolation-based methods.


\section{Introduction}

Reprogramming a cell to induce a desired cell fate or the identification of drug targets for disease treatment are examples of the multiple applications of the study of control mechanisms in biological systems. Mathematical modeling can help to predict potential control candidates in silico, which might reduce the need for the usually costly and time-consuming experimental testing \cite{sinergies}. Among the different mathematical frameworks, Boolean modeling stands out for its ability to capture the qualitative behavior and dynamics of biological systems even if there is a lack of detailed quantitative data. In the Boolean framework, each component is represented by a binary node that only admits two activity levels, 0 and 1, which might denote for instance in a gene regulatory network if a gene is active or not or if the concentration of a certain compound is above or below a certain threshold. The interactions between the components are described by logical functions. Despite its simplicity, Boolean modeling has been shown to reliantly capture the relevant dynamics of the modeled biological systems \cite{mapk_network, tlgl_network}. 

Control in Boolean networks is a broad field and many different approaches have been developed dealing with different scenarios and goals. Commonly, one wants to influence the system in such a way that the asymptotic dynamics fulfill the desired properties. Thus the focus is on introducing perturbations that influence the system attractors and their reachability properties. Some approaches aim at leading the system towards a desired attractor, from a certain initial state \cite{control_basins_seq} (source-target control) or from any possible initial state \cite{control_motifs} (full network control). We refer to this type of control as attractor control. In other cases, only the location of the attractors in a particular subspace is of importance, for instance if a certain phenotype defined by the values of a small set of marker components is the desired outcome. Several approaches have been developed to deal with this control problem, known as target control \cite{control_intervention_sets, target_control, control_trap_spaces, control_bcn}. 

Since control aims at manipulating dynamical properties of a model, it usually depends on the way the dynamics are derived from the Boolean function representing the system. The so-called synchronous update that updates every component in the model in each step gives rise to deterministic  dynamics and is highly amenable to computational analysis. Here, the only cycles in the dynamics are attractors and control can be interpreted as forcing all trajectories starting in an initial state to reach the desired target set. However, non-deterministic asynchronous updates that allow for different (if unknown) time delays of the possible component value changes have been shown to more realistically capture the behavior of biological systems. In this setting, there can be non-attractive cycles in the dynamics that trajectories can leave after entering. In application, trajectories that stay indefinitely in a non-attractive cycle are usually taken as modeling artifacts and are not further considered. Similarly, control approaches neglect such trajectories and solely aim at enforcing attractor properties. Several control methods have been developed specifically for synchronous \cite{control_algebra, control_multivalued_algebra} or asynchronous dynamics \cite{control_motifs, cabean} while others are applicable to any dynamics \cite{control_intervention_sets, control_trap_spaces, control_model_checking}.

Different types of model intervention targets can be considered. Most approaches use interventions that fix the state of a component to a certain value \cite{control_motifs, control_basins_seq, control_trap_spaces, control_model_checking}. This type of intervention, called node intervention or node control, can represent for instance the knockout or sustained activation of a gene in a gene-regulatory network. However, sometimes a certain node intervention might not be possible, either because it is unfeasible in practice or because the target component plays a potentially crucial role in some processes that should not be disrupted. In such cases, it is useful to consider interventions targeting only a specific interaction between two components, leaving the rest of the interactions unaltered. This type of intervention is known as edge intervention or edge control. Several approaches exist for identifying edge control strategies, dealing with target control in asynchronous \cite{control_bcn} or synchronous dynamics \cite{control_algebra, control_multivalued_algebra}.

Multiple approaches have been developed to identify and compute control strategies for Boolean networks in the different settings, using tools ranging from analysis of the stable motifs of the systems  \cite{control_motifs} to exploiting computational algebra methods \cite{control_algebra}. A core idea common to many of them is to utilize value percolation to test the effect of permanently fixing certain component values on the dynamical behavior. Methods based on value percolation can be implemented efficiently \cite{control_asp}. On the other hand, they are quite restrictive and might miss many possible control strategies. To bridge this gap, recent works have dealt with attractor control using basins of attraction \cite{cabean} or aimed at an exhaustive enumeration of all possible control strategies using model checking queries \cite{control_model_checking}. However, these approaches might entail high computational resources.

In order to benefit from the efficiency of value percolation and increase the number of identified control strategies, we explore the use of trap spaces for target control. Trap spaces are subspaces of the state space that are closed with respect to the dynamics. Consequently, a trap space contains at least one attactor. Trap spaces are in many cases good approximations of attractors \cite{klarner_attractor_approx} and can be efficiently computed for relatively large networks \cite{klarner_trap_spaces}. Trap spaces can be used as an intermediate control step since by leading the system to a trap space, one ensures that only attractors inside the trap space are reachable. Applying the usual percolation techniques to target trap spaces containing only desirable attractors can potentially uncover new control strategies for both node and edge control.

This article is an extended version of \cite{control_trap_spaces}, in which we developed a method for computing node control strategies utilizing value percolation in combination with trap space analysis for target control, potentially yielding richer solution sets while keeping computational efficiency.  We significantly broaden the theoretical and computational framework to include edge control, reworking the original material to obtain a consistent, comprehensive and flexible approach. For efficient implementation, we use a logical programing approach, namely Answer Set Programming (ASP), extending the works from \cite{control_asp} and \cite{control_asp_trap_spaces}. Finally, building on a case study from \cite{control_trap_spaces}, we show the applicability of the method and illustrate the potential inherent in comprehensive analysis using both node and edge control.

We start with a general overview about Boolean modeling (\Cref{Background}). Then, we introduce the different types of interventions considered in this work, node and edge, their effect on the controlled system and the theoretical basis for control strategy identification using value percolation and trap spaces (\Cref{CS}). The implementation of the method using Answer Set Programing is detailed in \Cref{Computation}. Finally, in \Cref{Application}, we show the applicability of our approach to a cell fate decision network.


\section{Background} \label{Background}

We define a \emph{Boolean network} on $n$ variables as a function $f \colon \B^n \rightarrow \B^n$, with $\B = \{0,1\}$. The set of variables or components $\{0, \dots n\}$ is denoted by $V$. The \emph{state space} of a Boolean function is denoted by $\B^n$ and every $x \in \B^n$ is a \emph{state} of the state space. We define the \emph{interaction graph} of a Boolean network $f$ as the labelled multi-digraph $(V,E)$ with $E \subseteq V \times V \times \{+,-\}$, admitting an edge from $i$ to $j$ if there exists $x \in \B^n$, such that $s = (f_j(\bar{x}^i) - f_j(x))(\bar{x}^i_i - x_i) \neq 0$, with $\bar{x}^i_k = 1 - x_k$ for $k = i$ and $\bar{x}^i_k = x_k$ for all $i \neq k$. The label of the edge is given by the sign of $s$. Thus, the interaction graph captures the activation (positive) and inhibition (negative) relations between the components of a Boolean network.

The dynamics of a Boolean network is defined by the \emph{state transition graph (STG)}, a directed graph with node set $\B^n$. Given a Boolean function $f$, we can define different dynamics depending on the way the components are updated, giving raise to different state transition graphs. For example, in the \emph{synchronous dynamics $SD(f)$} all the components that can be updated are updated at the same time, whereas in the \emph{asynchronous dynamics $AD(f)$} only one component is updated at a time. Thus, the synchronous state transition graph has an edge from a state $x$ to a state $y$ if $x \neq y$ and $f(x) = y$, whereas the asynchronous state transition graph has an edge from a state $x$ to a state $y$ if $f_i(x) = y_i \neq x_i$ for some $i \in V$ and $x_j = y_j$ for all $i \neq j \in V$. The \emph{generalized asynchronous dynamics $GD(f)$} includes transitions that update a set of components at a time. Consequently, the corresponding state transition graph has an edge from a state $x$ to a state $y$ if there exists $\emptyset \neq I \subseteq V$ such that $f_i(x) = y_i \neq x_i$ for $i \in I$ and $x_j = y_j$ for all $i \notin I$. In order to capture the different time scales that might coexist in a biological system, the asynchronous dynamics is often used. The work presented here is valid for any of the three dynamics introduced. We use $D(f)$ to refer to any of these dynamics.

The long term dynamics of the system is captured by the attractors. An \emph{attractor} is a minimal trap set, that is, a minimal set of states that is closed with respect to the dynamics. Attractors correspond to the terminal strongly connected components in the STG. An attractor $\A \subseteq \B^n$ is called \emph{steady state} when $|\A| = 1$ and \emph{cyclic attractor} when $|\A| > 1$. Steady states in a biological system might be associated with different cell fates or cell types and cyclic attractors with different cell cycles or cell processes with oscillatory behaviours.

Given a set of components $I \subseteq V$ and a state $d \in \B^n$, the \emph{subspace induced by $I$ and $d$} is defined as $\Sigma(I,d) = \{ x \in \B^n \ | \ x_i = d_i$ for all $i \in I \}$. We denote subspaces by writing the value $0$ or $1$ for the fixed variables and ${*}$ for the free ones. For example, the subspace ${*}{*}10$ denotes the set of states $\{ x \in \B^n \ | \ x_3 = 1$ and $x_4 = 0 \}$. We define the size of a subspace as the number of fixed variables. A subspace that is closed with respect to the dynamics is a \emph{trap space}. Trap spaces are invariant with respect to the type of update, contrary to attractors and trap sets, which might be different in different dynamics.


\section{Control strategies} \label{CS}

This work deals with target control and considers two types of interventions: node interventions and edge interventions. Node interventions fix a certain component to a certain value. More formally, a node intervention $(i,c)$, with $i \in V$ and $c \in \B$, sets the component $i$ and its regulatory function $f_i$ to the value $c$. This type of intervention can be seen, for example in a gene regulatory network, as the knock-out or permanent activation of a gene.

In the context of practical applications in biological systems, a node intervention is not always possible, for example if the component represents a gene or protein that is vital for other processes. In order to achieve the desired effect without altering the rest of the system, interventions acting only on the interaction between two components can be considered. This type of interventions, which would only alter the effect of a specific component on another without affecting the rest, are known as edge interventions. More formally, an edge intervention $(i,j,c)$, with $i,j \in V$ and $c \in \B$, fixes the value of the component $i$ in the regulatory function $f_j$ to the value $c$. By definition, when a regulatory function $f_j$ depends on a component $i$, there exists an edge from $i$ to $j$ in the interaction graph. An edge intervention can be seen as the deletion of such an edge in the interaction graph, since $f_j$ does not depend on the component $i$ after the intervention. In a biological system, an edge intervention can represent for example the modification of a protein that prevents it from binding to a certain component, while still allowing it to interact with the rest of the system. 

Given a Boolean network $f$ and a set of interventions $\mathcal{C}$, which might be node or edge interventions, we consider the simultaneous application of all interventions in $\mathcal{C}$ on $f$.
We write $f^{\mathcal{C}}$ for the function resulting from the application of the interventions in the set $\mathcal{C}$. In the first section we give the formal definition for the function $f^{\mathcal{C}}$ and a control strategy. In the second section, we recall properties of value percolation and establish the basis for control strategy identification for node and edge control.

\subsection{Controlled networks and control strategies}

We start by establishing the basic conditions that a set of interventions needs to satisfy in order to be consistent. These conditions aim at preventing, for instance, that a node intervention fixes a component to 1 while another is fixing the same component to 0.

Consider $\N \subseteq V \times \B$ and $\E \subseteq V \times V \times \B$. We call $\C = \N \cup \E$ a \emph{consistent} set of interventions if the following conditions are satisfied:
\begin{enumerate}[label=(\roman*)]
\item for all $i,j \in V$ and $c,c' \in \B$, if $(j,c) \in \C$, then $(i,j,c') \notin \C$;
\item for all $i,j \in V$ and $c,c' \in \B$, if $(i,c) \in \mathcal{C}$, then $(i,\bar{c}) \notin \mathcal{C}$ and $(i,j,c') \notin \mathcal{C}$;
\item for all $i,j \in V$ and $c \in \B$, if $(i,j,c) \in \mathcal{C}$, then $(i,j,\bar{c}) \notin \mathcal{C}$. 
\end{enumerate}

The first condition ensures that node and edge interventions do not act on the same target. The second guarantees that when a node intervention fixes the value of a component, no other intervention fixes that component. The last one prevents edge interventions from fixing the same component to different values in the same regulatory function.

In order to describe the effect of node and edge interventions in a Boolean network, we first define a function $h^{j,\mathcal{C}} \colon \B^n \rightarrow \B^n$ that, given a set of consistent interventions $\mathcal{C}$ and a component $j \in V$, captures the effect of fixing the components involved in the interventions acting on the regulatory function of $j$. For every $x \in \B^n$, $i,j \in V$, we set
$$
h^{j,\mathcal{C}}(x)_i = \left\{
\begin{array}{ll}
c & \text{if } (i,c) \in \mathcal{C} \text{ or } (i,j,c) \in \mathcal{C} \text{ for some } c \in \B, \\
x_i & \text{otherwise}.
\end{array}
\right.
$$

Note that $h^{j,\mathcal{C}}$ is well-defined when $\mathcal{C}$ is consistent. Given a Boolean network $f$ and a consistent set of interventions $\mathcal{C}$, we can now define the controlled network $f^{\mathcal{C}}$. For every $k \in V$,
$$
f_k^{\mathcal{C}} = \left\{
\begin{array}{ll}
c & \text{if } (k,c) \in \mathcal{C} \text{ for some } c \in \B, \\
f_k \circ h^{k, \mathcal{C}} & \text{otherwise}. \\
\end{array}
\right.
$$

The interventions considered in node control fix certain components (nodes) to certain values. Thus, if a set of interventions consists exclusively of node interventions ($\C = \mathcal{N}$), then it can be associated with a subspace $\Sigma(I,d)$, with $I \subseteq V$ and $d \in \B^n$ such that $(i,d_i) \in \mathcal{N}$ if and only if $i \in I$ \cite{control_trap_spaces}. When all interventions are edge interventions ($\C = \mathcal{E}$) the controlled regulatory function for any component $k$ is given by $f^{\mathcal{E}}_k = f_k \circ h^{k, \mathcal{E}}$.

Given a node intervention $(i,c)$, one could consider a set of edge interventions that fix the regulatory function $f_i$ to $c$. For instance, if a node $i$ has only one incoming edge from $j$ in the interaction graph of $f$, that is, the regulatory function satisfies $f_i(x) = x_j$ or $f_i(x) = \bar{x}_j$, the node intervention $(i,c)$ is equivalent in the long-term dynamics to the edge intervention $(j,i,c)$ or $(j,i,\bar{c})$ respectively. Note that a node intervention might have multiple equivalent sets of edge interventions. For example, if $f_i(x) = x_j \lor x_k$ for some $i,j,k \in V$, using either the edge intervention $(j,i,1)$ or the edge intervention $(k,i,1)$ would have the same long-term effect in the dynamics as the node intervention $(i,1)$.

\begin{definition}
\label{def:cs}
Given a Boolean network $f$ and a subspace $P \subseteq \B^n$, a set of interventions $\mathcal{C}$ is a \emph{control strategy for the target $P$} in $D(f)$ if $\A \subseteq P$ for any attractor $\A$ of $D(f^{\mathcal{C}})$.
\end{definition}

A set of interventions defines a control strategy for a given target when all the attractors of the controlled network are contained in the target. When considering only node control, since a set of node interventions $\mathcal{N}$ defines a subspace, a control strategy can also be identified with the subspace associated with $\N$ (see \cite{control_trap_spaces,control_model_checking}).

We define the size of a control strategy $\mathcal{C}$ as the number of interventions $|\mathcal{C}|$. In the case of $\E = \emptyset$, the number of interventions corresponds to the number of fixed variables. In practical applications, we are interested in intervention sets that are minimal with respect to inclusion. This is a natural approach when considering interaction sets that contain only node interventions or only edge interventions. For simplicity, in this work we use the same definition of minimality for intervention sets that mix edge and node interventions.
Depending on the context, the resources required to implement different interventions can vary, and more sophisticated objective functions might take these differences into account.

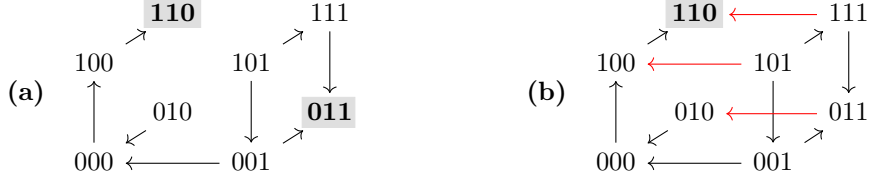
\begin{figure}[h!]
\centering
\begin{minipage}{0.05\linewidth}
\flushright
\textbf{(a)}
\end{minipage}
\begin{minipage}{0.38\linewidth}
\flushleft
\tikz[overlay]{
\filldraw[fill = gray!25, draw=gray!25] (1.2,0.88) rectangle (1.95,1.28);
\filldraw[fill = gray!25, draw=gray!25] (3.3,-0.45) rectangle (4,-0.05);}
\begin{tikzcd}[column sep=5, row sep=5]
 & \textbf{110} & & 111 \arrow[dd] \\
100 \arrow[ur] & & 101 \arrow[dd] \arrow[ru] & \\
 & 010 \arrow[ld] & & \textbf{011} \\
000 \arrow[uu] & & 001 \arrow[ll] \arrow[ru] \\
\end{tikzcd}
\end{minipage}
\begin{minipage}{0.05\linewidth}
\flushright
\textbf{(b)}
\end{minipage}
\begin{minipage}{0.38\linewidth}
\flushleft
\tikz[overlay]{
\filldraw[fill = gray!25, draw=gray!25] (1.2,0.88) rectangle (1.95,1.28);}
\begin{tikzcd}[column sep=5, row sep=5]
 & \textbf{110} & & 111 \arrow[dd] \arrow[ll, red] \\
100 \arrow[ur] & & 101 \arrow[ll, red] \arrow[dd] \arrow[ru] & \\
 & 010 \arrow[ld] & & 011 \arrow[ll, red] \\
000 \arrow[uu] & & 001 \arrow[ll] \arrow[ru] \\
\end{tikzcd}
\end{minipage}
\caption{(a) Asynchronous dynamics of the Boolean function $f(x) = (x_1 \bar{x}_3 \lor \bar{x}_2 \bar{x}_3$, $x_1 \lor x_3 $, $x_1 x_3 \lor x_2 x_3)$. (b) Asynchronous dynamics of the Boolean function $f^{\N}(x) = (x_1 \lor \bar{x}_2$, $x_1$, $0)$ with $\N = \{(3,0)\}$. Transitions that vary between $AD(f)$ and $AD(f^{\mathcal{N}})$ are marked in red. Attractors are marked in bold. $\N$ is a control strategy for $P = 110$ in $AD(f)$. $\B^n$ does not percolate to $P$ under $f^{\N}$ but percolates to the selected trap space $T$.}
\label{ex:nonperc}
\end{figure}

\begin{figure}[h!]
\centering
\begin{minipage}{0.05\linewidth}
\flushright
\textbf{(a)}
\end{minipage}
\begin{minipage}{0.38\linewidth}
\flushleft
\tikz[overlay]{
\filldraw[fill = gray!25, draw=gray!25] (1.3,0.88) rectangle (2,1.28);
\filldraw[fill = gray!25, draw=gray!25] (0.2,-1.1) rectangle (0.95,-0.7);
\filldraw[fill = gray!25, draw=gray!25] (3.35,0.88) rectangle (4.1,1.28);}
\begin{tikzcd}[column sep=5, row sep=5]
 & \textbf{110} & & \textbf{111} \\
100 \arrow[dd, red] & & 101 \arrow[ll] & \\
 & 010 \arrow[uu] \arrow[dl] & & 011 \arrow[ll] \arrow[uu] \arrow[dl] \\
\textbf{000} & & 001 \arrow[uu] \arrow[ll] \\
\end{tikzcd}
\end{minipage}
\begin{minipage}{0.05\linewidth}
\flushright
\textbf{(b)}
\end{minipage}
\begin{minipage}{0.38\linewidth}
\flushleft
\tikz[overlay]{
\filldraw[fill = gray!25, draw=gray!25] (1.3,0.88) rectangle (2,1.28);
\filldraw[fill = gray!25, draw=gray!25] (0.2, 0.23) rectangle (0.95,0.63);}
\begin{tikzcd}[column sep=5, row sep=5]
 & \textbf{110} & & 111 \arrow[ll, red]\\
\textbf{100} & & 101 \arrow[ll] & \\
 & 010 \arrow[uu] \arrow[dl] & & 011 \arrow[ll] \arrow[uu] \arrow[dl] \\
000 \arrow[uu, red] & & 001 \arrow[uu] \arrow[ll] \\
\end{tikzcd}
\end{minipage}
\caption{(a) Asynchronous dynamics of the Boolean function $f(x) = (x_2 \lor x_3$, $x_1 x_2$, $x_1 x_2 x_3)$. (b) Asynchronous dynamics of the Boolean function $f^{\mathcal{E}}(x) = (1, x_1 x_2, 0)$, with $\mathcal{E} = \{(2,1,1), (2,3,0)\}$. Transitions that vary between $AD(f)$ and $AD(f^{\mathcal{E}})$ are marked in red. Attractors are marked in bold. $\mathcal{E}$ is a control strategy for $P = 1{*}0$. There is no control strategy for $P$ consisting only of a node intervention on the second component.}
\label{ex:csedge}
\end{figure}
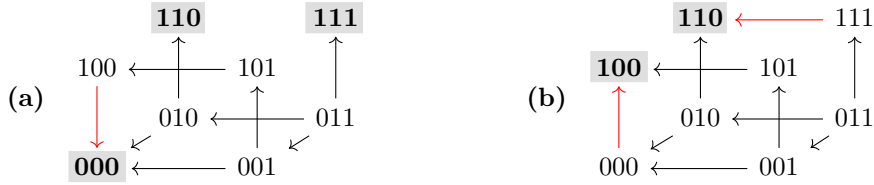

An example of a control strategy using node interventions is shown in \Cref{ex:nonperc}, where the set $\N = \{(3,0)\}$, associated with the subspace ${*}{*}0$, is a control strategy for the one-element target subspace $P = 110$, since $f^{\N}$ only has one attractor that is the steady state $110$. \Cref{ex:csedge} shows an example where the set of edge interventions $\mathcal{E} = \{(2,1,1),(2,3,0)\}$ is a control strategy for the target $P = 1{*}0$. Note that, if we do not allow interventions on the variables fixed in the target $P$, there are no node control strategies for $P$, since $000 \notin P$ is a steady state of $f^{\N_0}$, with $\N_0 = \{(2,0)\}$, and $111 \notin P$ is a steady state of $f^{\N_1}$, with $\N_1 = \{(2,1)\}$. Thus, in this scenario, control can only be achieved by using edge interventions. This example illustrates how edge interventions can broaden the possibilities for control.

\subsection{Value percolation and control strategy identification} \label{main_cs}

In the following, we recall the concept of value percolation and some properties of percolated subspaces and trap spaces that are helpful in the identification of control strategies.

Given a Boolean function $f$, we define the \emph{percolation function} with respect to $f$ as the function $F(f) \colon \mathcal{S} \rightarrow \mathcal{S}$, $S \mapsto F(f)(S)$, where $\mathcal{S}$ is the set of all subspaces in $\B^n$ and $F(f)(S)$ is the smallest subspace that contains $f(S)$ with respect to inclusion. That is, given a subspace $S \in \mathcal{S}$, $F(f)(S) = \Sigma(I,d)$ with $I = \{i\in V\ |\ |f_i(S)| = 1\}$ and $d \in \B^n$ such that $d_i = f_i(x)$ for all $x \in S$ for $i \in I$. Given two subspaces $S,S' \subseteq \B^n$, we say that the subspace \emph{$S$ percolates to $S'$ under $f$} if and only if there exists $k \geq 0$ such that $F(f)^k(S) = S'$.

For a set of interventions $\C$, $F(f^{\C})$ captures the propagation of the fixed values through the network. We can use the definition of percolation to formalise the notion of equivalence between intervention sets: we say that two intervention sets $\C_1$ and $\C_2$ are \emph{equivalent} if $F(f^{\C_1})(\B^n) = F(f^{\C_2})(\B^n)$.

Note that if $T$ is a trap space, $T' = F(f)(T)$ is also a trap space and $T' \subseteq T$. Moreover, for every $x \in T$ there exists a path in $D(f)$ from $x$ to some $y \in T'$.
Detailed proofs of these properties of subspace percolation can be found in \cite{control_model_checking}. A consequence of these observations is that, given a trap space $T$ that percolates to a subspace $S$, there cannot be an attractor $\A \subseteq T$ that is not contained in $S$, since for every state $x \in T$ there exists a path to some $y \in S$. Taking $T = \B^n$, we derive the following result.

\begin{proposition}
\label{prop_PO}
Let $P \subseteq \B^n$ be a subspace and $f$ a Boolean function. Let $\mathcal{C}$ be a set of interventions such that $\B^n$ percolates to $P$ under $f^{\mathcal{C}}$. Then $\mathcal{C}$ defines a control strategy in $D(f)$ for $P$.
\end{proposition}

We refer to the control strategies satisfying the conditions of \Cref{prop_PO} as \emph{control strategies by direct percolation}. An example of such a control strategy is shown in \Cref{ex:perc}. Several approaches to the identification of control strategies by direct percolation using node control have been developed \cite{control_intervention_sets, target_control} and there exist implementations that identify all control strategies by direct percolation efficiently \cite{Caspo}. However, there are still many control strategies that do not fulfill the conditions of \Cref{prop_PO}. \Cref{ex:nonperc} shows an example of control strategy that does not percolate to the target subspace.

\begin{figure}[h!]
\centering
\begin{minipage}{0.05\linewidth}
\flushright
\textbf{(a)}
\end{minipage}
\begin{minipage}{0.38\linewidth}
\flushleft
\tikz[overlay]{
\filldraw[fill = gray!25, draw=gray!25] (1.3,0.88) rectangle (2.05,1.28);
\filldraw[fill = gray!25, draw=gray!25] (0.2,-1.1) rectangle (0.95,-0.7);
\filldraw[fill = gray!25, draw=gray!25] (2.35,-1.1) rectangle (3.1,-0.7);}
\begin{tikzcd}[column sep=5, row sep=5]
 & \textbf{110} & & 111 \arrow[ll] \\
100 \arrow[ur] & & 101 \arrow[ll] \arrow[dd, red] \arrow[ru] & \\
 & 010 \arrow[ld, red] \arrow[uu] & & 011 \arrow[uu] \arrow[dl, red] \\
\textbf{000} & & \textbf{001} \\
\end{tikzcd}
\end{minipage}
\begin{minipage}{0.05\linewidth}
\flushright
\textbf{(b)}
\end{minipage}
\begin{minipage}{0.38\linewidth}
\flushleft
\tikz[overlay]{
\filldraw[fill = gray!25, draw=gray!25] (1.2,0.88) rectangle (1.95,1.28);}
\begin{tikzcd}[column sep=5, row sep=5]
 & \textbf{110} & & 111 \arrow[ll] \\
100 \arrow[ur] & & 101 \arrow[ll] \arrow[ru] & \\
 & 010 \arrow[uu] & & 011 \arrow[uu] \arrow[ll, red] \\
000 \arrow[uu, red] \arrow[ur, red] & & 001 \arrow[uu, red] \arrow[ur, red] \arrow[ll, red] \\
\end{tikzcd}
\end{minipage}
\caption{Asynchronous dynamics of the Boolean function $f(x) = (x_2 \lor x_1 \bar{x}_3$, $x_1$, $\bar{x}_1 x_3)$ (left) and $f^{\N}(x) = (1$, $1$, $0)$ with $\N = \{(1,1)\}$ (right). Transitions that vary between $AD(f)$ and $AD(f^{\mathcal{N}})$ are marked in red. Attractors are marked in bold. $\N$ is a control strategy for $P = 110$ in $AD(f)$. $\B^n$ percolates to $P$ under $f^{\N}$.}
\label{ex:perc}
\end{figure}
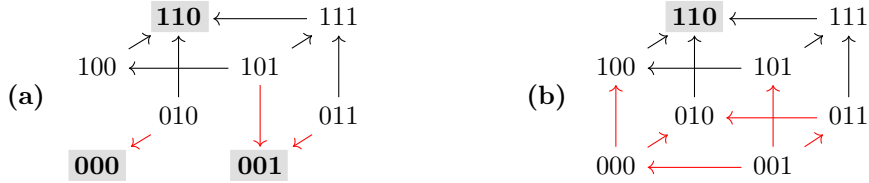

In order to exploit the efficiency of value percolation to identify more control strategies, we developed a method based on percolation that uses trap spaces \cite{control_trap_spaces}. As mentioned before, trap spaces are subspaces closed for the dynamics. Thus, each trap space contains at least one attractor. From all trap spaces of a Boolean function, we select the ones that contain only attractors belonging to the target subspace. We call such trap spaces \emph{selected trap spaces}. \Cref{prop_TS} introduces sufficient conditions for a subspace to be a control strategy for a target via a selected trap space.

\begin{proposition}
\label{prop_TS}
Let $P \subseteq \B^n$ be a subspace and $f$ a Boolean function. Let $T = \Sigma(I,d)$ be a trap space such that if $\A \subseteq T$ is an attractor of $D(f)$, then $\A \subseteq P$. Let $\mathcal{C}$ be a set of interventions such that $\B^n$ percolates to $T$ under $f^{\mathcal{C}}$ and for all $(i,c) \in \mathcal{C}$, $i \in I$ and for all $(i,j,c) \in \mathcal{C}$, $j \in I$. Then $\mathcal{C}$ defines a control strategy in $D(f)$ for $P$.
\end{proposition}

\begin{proof}
Let $\A \subseteq \B^n$ be an attractor for $D(f^{\mathcal{C}})$. Since $\B^n$ percolates to $T$ under $f^{\mathcal{C}}$, for every $x \in \B^n$, in particular for every $x \in \A$, there exists a path in $D(f^{\mathcal{C}})$ from $x$ to some $y \in T$. Therefore, $\A \subseteq T$. Since $\B^n$ percolates to $T$ under $f^{\mathcal{C}}$, $T$ is also a trap space in $f^{\mathcal{C}}$, so $f^{\mathcal{C}}_k(x) = d_k = f_k(x)$ for all $k \in I$ and $x \in T$. Since for all $(i,c) \in \mathcal{C}$, $i \in I$ and for all $(i,j,c) \in \mathcal{C}$, $j \in I$, we have that $f^{\mathcal{C}}_k(x) = f_k \circ h^{k,\mathcal{C}} (x) = f_k(x)$ for all $k \notin I$ and $x \in T$. Consequently, $f^{\mathcal{C}}(x) = f(x)$ for all $x \in T$. Since $\A \subseteq T$ and for all $x \in T, f(x) = f^{\mathcal{C}}(x)$, $\A$ is also an attractor of $D(f)$ and, therefore, $\A \subseteq P$.
\end{proof}

In the case $\mathcal{C} = \mathcal{N}$, the condition of $i \in I$ for all $(i,c) \in \mathcal{C}$ corresponds to $T \subseteq \Omega$, where $\Omega$ is the subspace associated to $\mathcal{N}$ \cite{control_trap_spaces}.

We call the control strategies satisfying the conditions of \Cref{prop_TS} \emph{control strategies via trap spaces}. \Cref{ex:nonperc} shows an example of this type of control strategy. $T = **0$ is a selected trap space for the target $P = 110$, since it contains only the attractor $\A = 110$. Consequently, $\N = \{(3,0)\}$ is a control strategy for $P$. Note that $\B^n$ does not percolate to $P$ under $f^{\N}$ but percolates to $T$.

We can easily identify all the selected trap spaces if the attractors of the Boolean network are known or, alternatively, if they can be approximated by minimal trap spaces \cite{klarner_attractor_approx}, that is, if each minimal trap space contains only one attractor and every attractor is included in a minimal trap space. Although attractor identification can be hard to achieve depending on the particular problem, the second property is easier to verify and is relatively common in Boolean networks modeling biological systems \cite{klarner_attractor_approx}.

Control strategies by direct percolation do not depend on the update. On the other hand, the selected trap spaces are defined in terms of the attractors, which might vary in different updates. As a consequence, control strategies via trap spaces are also in general update-dependent. This provides the method with enough flexibility to identify control strategies that are valid in one update but not in another.

A further advantage of the control strategies identified by \Cref{prop_TS} is that they allow for the control interventions to be eventually released. Once a selected trap space is reached, the system will remain in the trap space, regardless of whether the control interventions are active or not. This additional property widens the range of possible choices for control since interventions relying on agents that decay over time could also be considered.


\section{Methods: control strategy computation} \label{Computation}

The methods for control strategy identification presented in this work are based on the identification of sets of interventions that cause the state space to percolate either to the target subspace or to one of the selected trap spaces under the controlled function. Identifying all the minimal control strategies of this type might entail the exploration of all possible sets of interventions, whose number grows exponentially with the size of the network (for node control) or with the number of edges (for edge control).

The use of Answer Set Programming (ASP) was proposed by Kaminski et al. \cite{control_asp} to deal with the combinatorial explosion associated with node control. Answer Set Programming is a form of declarative programming that works well with hard combinatorial, search and optimization problems. This type of problems often entail a decision-making process over a set of candidates to decide whether they satisfy a specified constraint and possibly identify an optimised output. In order to solve a problem with ASP, one needs to provide a description of the problem using logical rules. Solving the original problem is then reduced to identifying the solutions of its corresponding logic program.

In \cite{control_asp_trap_spaces} we extended the work done in \cite{control_asp} to identify the control strategies presented in \cite{control_trap_spaces}. In this section, we recall the implementation of \cite{control_asp_trap_spaces} and extend it to deal with edge control.

\subsection{Problem encoding}

\definecolor{codegreen}{rgb}{0,0.6,0}
\definecolor{codegray}{rgb}{0.5,0.5,0.5}
\definecolor{codepurple}{rgb}{0.58,0,0.82}
\definecolor{backcolour}{rgb}{0.95,0.95,0.92}

\lstdefinestyle{mystyle}{
    backgroundcolor=\color{gray!10},   
    numberstyle=\tiny\color{codegray},
    basicstyle=\ttfamily\scriptsize,
    breakatwhitespace=false,         
    breaklines=true,                 
    captionpos=t,                    
    keepspaces=true,                 
    numbers=left,                    
    numbersep=5pt,                  
    showspaces=false,                
    showstringspaces=false,
    showtabs=false,                  
    tabsize=2
}

The ASP encoding for control strategy identification consist of two parts: the encoding of the control problem (program instance) and the encoding of the computation process (main program). The encoding of the control problem includes the Boolean function, the target subspaces and selected trap spaces, the limit size of the control strategies and the restrictions on the nodes and edges that can be used for control.

The Boolean function is encoded from its complete disjuctive normal form (DNF), which consists of the disjunction of its prime implicants, as described in \cite{Caspo}. Every component of the network is declared in the literal {\ttfamily variable} (line 1). We allow the possibility of excluding certain interventions from the control candidates, for instance in the case that an intervention is not feasible for application. The node and edge interventions that we want to exclude from the control are declared in the literals {\ttfamily avoid$\_$node} or {\ttfamily avoid$\_$edge} respectively (line 2). The Boolean network from the example in \Cref{ex:nonperc} is encoded as follows. The literal {\ttfamily formula} (line 4) links every variable with its DNF, described by the literals {\ttfamily dnf} and {\ttfamily clause} (lines 5-8). The regulatory function of the first component $f_1(x) = x_1 \bar{x}_3 \lor \bar{x}_2\bar{x}_3$ is declared in the literal {\ttfamily formula(x1,0)} (line 4) and linked to its two clauses {\ttfamily dnf(0,0)} and {\ttfamily dnf(0,1)} (line 5). The first clause $x_1 \bar{x}_3$ is encoded in the literals {\ttfamily clause(0, x1, 1)} and {\ttfamily clause(0, x3, -1)} (line 6). Note that we use {\ttfamily -1} and {\ttfamily 1} in the third variable of the literal {\ttfamily clause} to denote whether a variable is negated or not, respectively. To ease the encoding, we also use the value {\ttfamily -1} to represent the Boolean value $0$ in the rest of the program.

\begin{lstlisting}[style=mystyle, escapechar=!]
 variable(x1). variable(x2). variable(x3). 
 avoid_node(x1). avoid_edge(x1,x1). avoid_edge(x1,x2). avoid_edge(x1,x3). 

 formula(x1, 0). formula(x2, 1). formula(x3, 2). 
 dnf(0, 0). dnf(0, 1). dnf(1, 2). dnf(1, 3). dnf(2, 4). dnf(2, 5). 
 clause(0, x1, 1). clause(0, x3, -1). clause(1, x2, -1). clause(1, x3, -1). 
 clause(2, x3, 1). clause(3, x1, 1). clause(4, x2, 1). clause(4, x3, 1).
 clause(5, x1, 1). clause(5, x3, 1). 
\end{lstlisting} 

The target subspace and the target trap spaces are encoded in the literal {\ttfamily subspace} (line 9). We use two types of identifier: positive and negative. The positive identifier marks the subspace as a selected trap space and the negative identifier marks it as the direct target. The fixed variables of each subspace are encoded in the variable {\ttfamily goal} (line 11) as in \cite{control_asp} and \cite{control_asp_trap_spaces}. A limit size on the number of interventions is set in line 13.

\begin{lstlisting}[style=mystyle, escapechar=!, firstnumber=last]
 subspace(-1). subspace(0). subspace(1). 
 goal(-1, x3, -1). goal(-1, x2, 1). goal(-1, x1, 1).
 goal(0, x3, -1). goal(0, x2, 1). goal(0, x1, 1). goal(1, x3, -1). 

 #const maxsize=3.
\end{lstlisting}

The main ASP program for control strategy identification can be divided in four parts: candidate instantiation, new controlled function instantiation (only necessary for edge interventions), percolation step and satisfaction requirements. The first two parts differ in node and edge control whereas the last two are the same. In the following, we describe each step in detail.

The candidate instantiation for node control is adapted from \cite{control_asp} as described in \cite{control_asp_trap_spaces}. Line 6 is required for control via trap spaces. When the state space percolates to a selected trap space, it ensures that the candidate interventions are chosen among the variables fixed in the trap space, as required in \Cref{prop_TS}.

\begin{lstlisting}[style=mystyle, escapechar=!, label={asp_candidates_node}]
 goal(T,S) :- goal(Z,T,S), Z < 0.
 satisfy(V,W,S) :- formula(W,D); dnf(D,C); clause(C,V,S).
 closure(V,T)   :- goal(V,T).
 closure(V,S*T) :- closure(W,T); satisfy(V,W,S); not goal(V,-S*T).
 { node(V,S) : closure(V,S), candidate(V), satisfied(Z), Z < 0 }.
 { node(V,S) : goal(Z,V,S), candidate(V), satisfied(Z), Z >=0 }.
 :- node(V,1); node(V,-1).
 node(V) :- node(V,S).
\end{lstlisting}

The candidate instantiation for edge control is shown below. First, the candidate edge interventions are generated (lines 9-10). Note that we only consider interventions $(i, j, c)$ such that $i$ regulates $j$ and that we exclude forbidden edges. As in node control, we exclude contradictory interventions (line 11). We keep track of the controlled edges in the variable {\ttfamily edge(Vi,Vj)} (line 12).

\begin{lstlisting}[style=mystyle, escapechar=!,firstnumber=last]
 { edge(Vi,Vj,1) } :- formula(Vj,D), dnf(D,C), clause(C,Vi,S), not avoid_edge(Vi,Vj).
 { edge(Vi,Vj,-1) } :- formula(Vj,D), dnf(D,C), clause(C,Vi,S), not avoid_edge(Vi,Vj).
 :- edge(Vi,Vj,1), edge(Vi,Vj,-1).
 edge(Vi,Vj) :- edge(Vi,Vj,S).
\end{lstlisting}

When considering node and edge interventions together, the following restrictions are also added, to ensure that the set of interventions is consistent.

\begin{lstlisting}[style=mystyle, escapechar=!,firstnumber=last]
 :- node(V), edge(V,Vj).
 :- node(V), edge(Vi,V).
\end{lstlisting}

The effect of the edge interventions in the Boolean function is captured before the percolation step. Note that this step is not needed for node control. The literals {\ttfamily new$\_$clause}, {\ttfamily new$\_$dnf}, {\ttfamily new$\_$formula} are instantiated for every term, clause and DNF respectively that are not fixed by edge interventions. Thus, these literals represent the DNF of the controlled Boolean function. In particular, {\ttfamily new$\_$clause(C,V,S)} is instantiated for every term that is not affected by an edge intervention (line 15). {\ttfamily remove$\_$dnf} captures the clauses of the DNF that are set to 0 by edge interventions and, consequently, are not part of the controlled function (line 16). {\ttfamily new$\_$dnf} is then instantiated for the rest of the DNF clauses (line 17). In a similar way, {\ttfamily remove$\_$formula} captures the regulatory functions that have a clause in the DNF set to 1 by edge interventions, since they become the constant 1 and do not have an associated DNF (line 18). In this case, the literal {\ttfamily fixed$\_$node(V,1)} is also instantiated to indicate that the component {\ttfamily V} is set to 1 (line 21). The remaning regulatory functions with at least one DNF clause different from 0 are captured in {\ttfamily new$\_$formula} (line 19). If all DNF clauses are set to 0 by edge interventions, the regulatory function becomes the constant 0 and the literal {\ttfamily fixed$\_$node(V,-1)} is instantiated to indicate that the component {\ttfamily V} is set to 0 (line 22). 

\begin{lstlisting}[style=mystyle, escapechar=!,firstnumber=last]
 new_clause(C,V,S) :- clause(C,V,S); dnf(D,C); formula(Vj,D); not edge(V,Vj).
 remove_dnf(D,C):- clause(C,Vi,-S); edge(Vi,Vj,S); dnf(D,C); formula(Vj,D).
 new_dnf(D,C) :- new_clause(C,Vi,S); dnf(D,C); formula(Vj,D); not remove_dnf(D,C).
 remove_formula(Vj,D) :- dnf(D,C); formula(Vj,D); edge(Vi,Vj,S) : clause(C,Vi,S).
 new_formula(V,D) :- new_dnf(D,C); formula(V,D); not remove_formula(V,D).

 fixed_node(V,1) :- remove_formula(V,D).
 fixed_node(V,-1) :- not remove_formula(V,D); not new_formula(V,D); formula(V,D).
\end{lstlisting}

The regulatory functions that become constants either through node or edge control are captured in the literals {\ttfamily intervention(V,S)} and {\ttfamily intervention(V)}.

\begin{lstlisting}[style=mystyle, escapechar=!,firstnumber=last]
intervention(V,S) :- node(V,S).
intervention(V,S) :- not node(V,S), not node(V,-S), fixed_node(V,S).
intervention(V) :- intervention(V,S).
\end{lstlisting}

The percolation effect is then encoded in the same way as described in \cite{control_asp} using the literal {\ttfamily intervention(V,S)} (lines 26-30).

\begin{lstlisting}[style=mystyle, escapechar=!, firstnumber=last]
 eval_formula(Z,V,S) :- subspace(Z); intervention(V,S).
 free(Z,V,D) :- new_formula(V,D); subspace(Z); not intervention(V).
 eval_clause(Z,C,-1) :- new_clause(C,V,S); eval_formula(Z,V,-S).
 eval_formula(Z,V,1) :- free(Z,V,D); eval_formula(Z,W,T) : new_clause(C,W,T); new_dnf(D,C).
 eval_formula(Z,V,-1) :- free(Z,V,D); eval_clause(Z,C,-1) : new_dnf(D,C).
\end{lstlisting}

Finally, we ensure that a candidate subspace is a control strategy as described in \cite{control_asp_trap_spaces}, by requiring that at least one subspace constraint is satisfied (lines 31-33). A limitation on the number of interventions is also added (line 34). 

\begin{lstlisting}[style=mystyle, escapechar=!, label={asp_requierements},firstnumber=last]
 not satisfied(Z) :- goal(Z,T,S), not eval_formula(Z,T,S), subspace(Z).
 satisfied(Z) :- eval_formula(Z,T,S) : goal(Z,T,S); subspace(Z).
 0 < { satisfied(Z) : subspace(Z) }.
 :- maxsize > 0; maxsize + 1 { node(V,R); edge(Vi,Vj,S) }.
\end{lstlisting}

\subsection{Main algorithm} \label{sub:algo}

The algorithm for control strategy identification is detailed in \Cref{alg:ca}. It takes as inputs the Boolean function $f$, the target subspace $P$, the type of control method $method$, the limit size for the control strategies $m$ and, optionally, the list of attractors $attr$ (line 1). The Boolean function, the target subspace, the selected trap spaces and the limit size are used as input for the ASP program described above in \emph{createCandidatesAndPercolate}, which computes the corresponding control strategies. The Boolean function is given as a list of prime implicants so that it can be directly encoded as a complete DNF.

\begin{algorithm}[H]
\renewcommand{\thealgorithm}{1}
\caption{Control strategies for a target subspace}\label{alg:ca}
\begin{algorithmic}[1]
\Function{ControlStrategies}{$f$, $P$, $method$, $m$, $attr$}
	\If {method = ``direct percolation''}:
		\State \textbf{CS} $\gets$ createCandidatesAndPercolate($f$, $P$, $-$, $m$)
	\Else:
		\State \textbf{T} $\gets$ trapSpaces($f$)
		\State \textbf{selTS} $\gets$ selectedTrapSpacesType1(\textbf{T}, $P$)
		\If {$attr \neq \emptyset$}:
			\State \textbf{selTS} $\gets$ \textbf{selTS} + selectedTrapSpacesType2(\textbf{T}, $P$, $attr$)
		\EndIf
	\EndIf
	\If {method = ``trap spaces''}:
	\State \textbf{CS} $\gets$ createCandidatesAndPercolate($f$, $-$, \textbf{selTS}, $m$)
	\EndIf
	\If {method = ``combined''}:
		\State \textbf{CS} $\gets$ createCandidatesAndPercolate($f$, $P$, \textbf{selTS}, $m$)
	\EndIf
	\State \Return \textbf{CS}
\EndFunction
\end{algorithmic}
\end{algorithm}

\Cref{alg:ca} allows for the computation of control strategies by direct percolation (lines 2-3), via the trap spaces method (lines 9-10) and using the two methods combined, meaning that both percolation to the target subspace and selected trap spaces is considered (lines 11-12). When searching for control strategies via the trap spaces, we distinguish two types of selected trap spaces: trap spaces contained in $P$ (Type 1) (line 6) and trap spaces not contained in $P$ but containing only attractors in $P$ (Type 2) (line 8). Note that selected trap spaces of Type 2 are only identified when all the attractors are known or can be approximated by minimal trap spaces (line 7). Moreover, in order to avoid unnecessary calculations, we only consider non-percolating trap spaces, that is, trap spaces that do not  percolate to smaller ones, since all the subspaces percolating to a trap space $T$ also percolate to $F(f)(T)$.

We implemented \Cref{alg:ca} using PyBoolNet \cite{PyBoolNet}, a Python package for the generation, modification and analysis of Boolean networks. PyBoolNet also provides an efficient computation of trap spaces for relatively large networks, which we use for the computation of the selected trap spaces, and a method to check whether the attractors of a Boolean network can be approximated by minimal trap spaces \cite{klarner_attractor_approx}. To solve the ASP problem, we use \emph{clingo}, developed by Potassco, the Potsdam Answer Set Solving Collection \cite{clingo}. The source code of the implementation of \Cref{alg:ca} is available at \url{https://github.com/Lauracf/trap-space-control}.


\section{Results} \label{Application}

In this section, we consider the applicability of our method to a cell fate decision model. We study the network introduced by Grieco et al. (2013) \cite{mapk_network} to model the effect of the Mitogen-Activated Protein Kinase (MAPK) pathway on cell fate decisions in bladder cancer cells (see \Cref{fig:mapk}). The network consists of 53 Boolean variables, including the four inputs DNA-damage, EGFR-stimulus, FGFR3-stimulus and TGFBR-stimulus. The states of the three outputs of the network (Apoptosis, Growth-Arrest and Proliferation) indicate the enablement or disablement of the corresponding processes that represent the different cell fates or phenotypes considered in \cite{mapk_network}.

\begin{figure}[h!]
\centering
\includegraphics[width=0.95\textwidth]{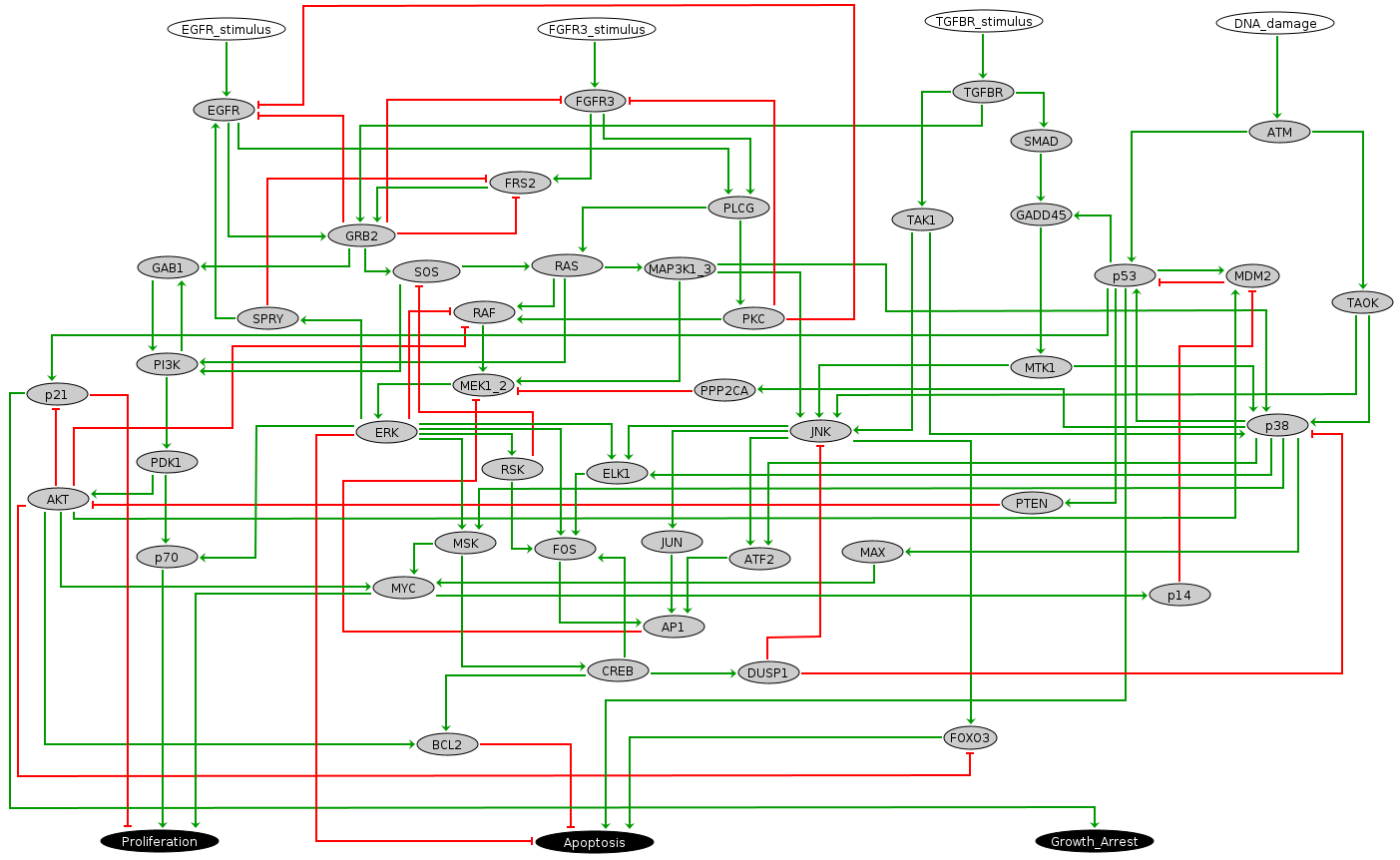}
\caption{MAPK network presented in \cite{mapk_network}. Figure obtained using GINsim \cite{GINsim}. Input and output nodes are colored in white and black respectively. Green edges denote activations and red edges denote inhibitions.}
\label{fig:mapk}
\end{figure}

There are 18 attractors in the asynchronous dynamics, of which 12 are steady states and 6 are cyclic. The attractors are in one-to-one correspondence with the minimal trap spaces, that is, each attractor is contained in a minimal trap space and each minimal trap space only contains one attractor \cite{klarner_attractor_approx}. Therefore, we can use the selected trap spaces of Type 1 and Type 2 (see \Cref{sub:algo}) to search for control strategies via trap spaces.

In the first part, we target the subspace defined by the apoptotic phenotype and compare the control strategies identified via trap spaces to the ones by direct percolation, first for node control and then for edge control. In the second part, we consider the attractors of the asynchronous dynamics by targeting the minimal trap spaces. We compare the control strategies identified by direct percolation, via trap spaces and with the combination of the two methods (see \Cref{alg:ca}) for four steady states for node and edge control. In all the cases, we obtain new control strategies via trap spaces missed by direct percolation.

\subsection{Target: apoptotic phenotype}

We start by considering as target the apoptotic phenotype that is defined by the subspace obtained by fixing Apoptosis to 1, Growth-Arrest to 1 and Proliferation to 0 as in \cite{mapk_network}. We refer to this subspace as the apoptotic target. We identify 103 non-percolating selected trap spaces. We set a limit size of three interventions.

In this setting, we obtain 271 control strategies for node control up to size 3: three of size 1, 106 of size 2, 162 of size 3. The three control strategies of size 1 are $\{$(TGFBR-stimulus, 1)$\}$, $\{$(TGFBR, 1)$\}$ and $\{$(DNA-damage, 1)$\}$, the last one being obtained only via trap spaces. Under the control strategy $\{$(DNA-damage, 1)$\}$, the state space percolates to the trap space $\{$(ATM, 1), (DNA-damage, 1), (TAOK, 1)$\}$, which contains only attractors in the apoptotic target. This minimal control strategy is not identified by direct percolation. \Cref{tab:mapknode} shows the number of control strategies identified by each method. Note that the number of control strategies of size 2 and 3 is lower when the methods are used in combination, since some of the control strategies identified by direct percolation are non-minimal and they are included in the control strategy of size 1 that is not detected.

\begin{table}[h!]
\caption{Number and size of the control strategies identified by the different methods up to size 3 for the apoptotic target using node control.}
\adjustbox{scale=0.84, center}{
\begin{tabu}{|l|[2pt]c|c|c|}
\hline
 & $| \mathcal{N} | = 1$ & $| \mathcal{N} | = 2$ & $| \mathcal{N} | = 3$ \\
\hline
By direct percolation & 2 & 124 & 175 \\
Via trap spaces & 2 & 0 & 0 \\
\hline
Combined & 3 & 106 & 162 \\
\hline
\end{tabu}
}
\label{tab:mapknode}
\end{table}

Using edge control we obtain 950 control strategies up to size 3: three of size 1, 117 of size 2 and 830 of size 3. \Cref{tab:mapkedge} shows the number of edge control strategies identified by each method. The three edge control strategies of size 1 are equivalent to the node interventions identified as control strategies of size 1. This results from the three variables involved in the size 1 node control strategies having a unique incoming edge. For example (TGFBR-stimulus, TGFBR, 1) has exactly the same effect as (TGFBR, 1), since TGFBR is uniquely regulated by TGFBR-stimulus.

\begin{table}[h!]
\caption{Number and size of the control strategies identified by the different methods up to size 3 for the apoptotic target using edge control.}
\adjustbox{scale=0.84, center}{
\begin{tabu}{|l|[2pt]c|c|c|}
\hline
 & $| \mathcal{E} | = 1$ & $| \mathcal{E} | = 2$ & $| \mathcal{E} | = 3$ \\
\hline
By direct percolation & 2 & 137 & 893 \\
Via trap spaces & 2 & 0 & 0 \\
\hline
Combined & 3 & 117 & 830 \\
\hline
\end{tabu}
}
\label{tab:mapkedge}
\end{table}

In other cases, edge control allows intervention strategies that would be too restrictive in node control. For example, the two edge interventions (MAP3K1-3, p38, 1) and (MSK, CREB, 0), which fix the activation of MAP3K1-3 in p38 and the inhibition of MSK in CREB, lead the controlled system to percolate to the apoptotic target. However, fixing MAP3K1-3 to 1 and MSK to 0 does not, since the controlled system displays non-apoptotic steady states, which are not present in the original dynamics.

Allowing the combination of node and edge interventions, we obtain over three thousand control strategies up to size 3. Note that these include all the control strategies obtained for node and edge control. In particular, the six control strategies of size 1 correspond to the three control strategies of node control and the three of edge control. The number of control strategies identified by each method is shown in \Cref{tab:mapkcomb}.

\begin{table}[h!]
\caption{Number and size of the control strategies identified by the different methods up to size 3 for the apoptotic target combining node and edge control.}
\adjustbox{scale=0.84, center}{
\begin{tabu}{|l|[2pt]c|c|c|}
\hline
 & $| \mathcal{C} | = 1$ & $| \mathcal{C} | = 2$ & $| \mathcal{C} | = 3$ \\
\hline
By direct percolation & 4 & 530 & 3569 \\
Via trap spaces & 4 & 0 & 0 \\
\hline
Combined & 6 & 454 & 3299 \\
\hline
\end{tabu}
}
\label{tab:mapkcomb}
\end{table}

We observe that there are many control strategies that mix node and edge interventions. Most of them include interventions already appearing in control strategies consisting exclusively of node interventions or of edge interventions. In some cases, we find mixed control strategies that are equivalent to a node control strategy or an edge control strategy where a node intervention is substituted by an equivalent edge intervention or vice versa. For example, the control strategy $\{$(CREB, DUSP1, 0), (TAOK, 1)$\}$ is equivalent to the control strategy $\{$(CREB, DUSP1, 0), (ATM, TAOK, 1)$\}$, since the node intervention (TAOK, 1) is equivalent to the edge intervention (ATM, TAOK, 1). There are also control strategies involving interventions that are not part of any node or edge control strategy. This is the case for 
$\{$(FGFR3, FRS2, 1), (GRB2, FRS2, 0), (p38, 1)$\}$, where neither (FGFR3, FRS2, 1) nor (GRB2, FRS2, 0) appear in any edge control strategy.

\subsection{Target: minimal trap spaces}

When computing control strategies for the minimal trap spaces, the input components need to be fixed in order to ensure that their value matches the one fixed in the target. Since each input combination identifies a separate trap space, there is at least one attractor per input combination. There are sixteen possible input combinations, fourteen of which identify subspaces that contain a unique attractor. These input combinations therefore give minimal node control strategies for the corresponding attractors. The subspaces induced by the two remaining input combinations ($\mathcal{S}_1 = \{$EGFR-stimulus = 0, FGFR3-stimulus = 0, TGFBR-stimulus = 0 and DNA-damage = 0$\}$, $\mathcal{S}_2 = \{$EGFR-stimulus = 0, FGFR3-stimulus = 0, TGFBR-stimulus = 0 and DNA-damage = 1$\}$) contain two steady states each and, therefore, further control interventions are needed. \Cref{tab:mapknode_ss}, \Cref{tab:mapkedge_ss} and \Cref{tab:mapkcomb_ss} show the number and size of the control strategies up to size 7 (the number of inputs plus three) of these four steady states for node, edge and mixed control respectively. Note that in all the cases there are control strategies identified via trap spaces not captured by direct percolation and there is no minimal control strategy identified by direct percolation missed via trap spaces. Moreover, no control strategy of size 5 is found for direct percolation for any of the steady states.

\begin{table}[h!]
\caption{Number and size of the control strategies identified by the different methods up to size 7 for the different steady states using node control. Note that there is no control strategy up to size 4. $s_1$ and $s'_1$ denote the two steady states in $\mathcal{S}_1$ and $s_2$ and $s'_2$ the two steady states in $\mathcal{S}_2$.}
\adjustbox{scale=0.84, center}{
\begin{tabu}{|l|[2pt]c|c|c|[2pt]c|c|c|[2pt]c|c|c|[2pt]c|c|c|}
\hline
 & \multicolumn{3}{c|[2pt]}{$s_1$} & \multicolumn{3}{c|[2pt]}{$s'_1$} & \multicolumn{3}{c|[2pt]}{$s_2$} & \multicolumn{3}{c|}{$s'_2$} \\
\hline
$| \N |$ & 5 & 6 & 7 & 5 & 6 & 7 & 5 & 6 & 7 & 5 & 6 & 7 \\
\hline
\cline{1-13}
\cline{1-13}
By direct percolation & 0 & 0 & 60 & 0 & 0 & 32 & 0 & 14 & 2 & 0 & 14 & 2 \\
Via trap spaces & 2 & 0 & 0 & 0 & 8 & 12 & 0 & 22 & 14 & 2 & 0 & 0 \\
\hline
Combined & 2 & 0 & 0 & 0 & 8 & 12 & 0 & 22 & 14 & 2 & 0 & 0 \\
\hline
\end{tabu}
}
\label{tab:mapknode_ss}
\end{table}

\begin{table}[h!]
\caption{Number and size of the control strategies identified by the different methods up to size 7 for the different steady states using edge control. Note that there is no control strategy up to size 4. $s_1$ and $s'_1$ denote the two steady states in $\mathcal{S}_1$ and $s_2$ and $s'_2$ the two steady states in $\mathcal{S}_2$.}
\adjustbox{scale=0.84, center}{
\begin{tabu}{|l|[2pt]c|c|c|[2pt]c|c|c|[2pt]c|c|c|[2pt]c|c|c|}
\hline
 & \multicolumn{3}{c|[2pt]}{$s_1$} & \multicolumn{3}{c|[2pt]}{$s'_1$} & \multicolumn{3}{c|[2pt]}{$s_2$} & \multicolumn{3}{c|}{$s'_2$} \\
\hline
$| \E |$ & 5 & 6 & 7 & 5 & 6 & 7 & 5 & 6 & 7 & 5 & 6 & 7 \\
\hline
\cline{1-13}
\cline{1-13}
By direct percolation & 0 & 0 & 150 & 0 & 0 & 84 & 0 & 22 & 58 & 0 & 33 & 50 \\
Via trap spaces & 6 & 11 & 157 & 0 & 16 & 168 & 0 & 50 & 72 & 6 & 33 & 40 \\
\hline
Combined & 6 & 11 & 157 & 0 & 16 & 168 & 0 & 50 & 72 & 6 & 33 & 40 \\
\hline
\end{tabu}
}
\label{tab:mapkedge_ss}
\end{table}

\begin{table}[h!]
\caption{Number and size of the control strategies identified by the different methods up to size 7 for the different steady states combining node and edge control. Note that there is no control strategy up to size 4. $s_1$ and $s'_1$ denote the two steady states in $\mathcal{S}_1$ and $s_2$ and $s'_2$ the two steady states in $\mathcal{S}_2$.}
\adjustbox{scale=0.84, center}{
\begin{tabu}{|l|[2pt]c|c|c|[2pt]c|c|c|[2pt]c|c|c|[2pt]c|c|c|}
\hline
& \multicolumn{3}{c|[2pt]}{$s_1$} & \multicolumn{3}{c|[2pt]}{$s'_1$} & \multicolumn{3}{c|[2pt]}{$s_2$} & \multicolumn{3}{c|}{$s'_2$} \\
\hline
$| \C |$ & 5 & 6 & 7 & 5 & 6 & 7 & 5 & 6 & 7 & 5 & 6 & 7 \\
\hline
\cline{1-13}
\cline{1-13}
By direct percolation & 0 & 0 & 12720 & 0 & 0 & 7040 & 0 & 1168 & 2608 & 0 & 1440 & 2048 \\
Via trap spaces & 128 & 176 & 2488 & 0 & 768 & 6848 & 0 & 2320 & 4528 & 128 & 528 & 616 \\
\hline
Combined & 128 & 176 & 2488 & 0 & 768 & 6848 & 0 & 2320 & 4528 & 128 & 528 & 616 \\
\hline
\end{tabu}
}
\label{tab:mapkcomb_ss}
\end{table}

As in the case of the apoptotic target, there are edge control strategies allowing interventions that would not be possible using only node control. For example fixing the component GRB2 either to 0 or to 1, together with the corresponding input interventions, does not lead to a system with $s_1$ as the unique attractor. However, fixing GRB2 in the edge intervention (GRB2, GAB1, 1) in addition to the input interventions leads to a controlled dynamics that has $s_1$ as a unique attractor.

When considering mixed interventions, in contrast to the apoptotic target case, all the interventions appearing in minimal control strategies also occur in some strategy composed exclusively of node or exclusively of edge interventions. As can be seen from the numbers in \Cref{tab:mapkcomb_ss}, we still gain many mixed control strategies and thus more flexibility for choosing interventions that are both realizable in the lab and as non-invasive as possible for the system.

\subsection{Running times}

All the results presented here were obtained with a regular desktop 8-processor computer, Intel\textsuperscript{\textregistered}Core\textsuperscript{\tiny TM} i7-2600 CPU at 3.40GHz, 16GB memory. The running times of the control strategy computation for each method, target and type of control are shown in \Cref{tab:times}. These refer to the total times needed for \Cref{alg:ca} to terminate, including the computation of the selected trap spaces when needed. We can observe how the number of candidate interventions affects the time required for each method. Node control is the fastest, around a few centiseconds, whereas edge control requires a few seconds. The running times of the different methods vary from a few seconds to a few minutes when combining the two types of interventions. Although the apoptotic phenotype is the target with the highest number of selected trap spaces, we do not observe a significant increase of the running time with respect to the steady states. 
This could perhaps result from the additional constraint on candidate interventions that is imposed when working with selected trap spaces, requiring the interventions to be selected among the variables fixed in the trap space.

\begin{table}[h!]
\caption{Running times (in seconds) for the control strategy computation targeting the apoptotic phentoype and the four steady states in the MAPK network.}
\adjustbox{scale=0.84, center}{
\begin{tabular}{|c|c|c|c|c|c|c|}
\hline
\multirow{2}*{Target} & \multirow{2}*{Target size} & Number of selected & \multirow{2}*{Method} & \multicolumn{3}{c|}{Time (s)} \\
\cline{5-7}
& & trap spaces & & Node & Edge & Both \\
\hline
\multirow{3}*{Apoptotic phenotype} & \multirow{3}*{3} & \multirow{3}*{103} & By direct percolation & 0.18 & 7.05 & 280.18 \\
& & & Via trap spaces & 0.47 & 0.82 & 0.74 \\
& & & Combined & 7.15 & 117.36 & 461.37 \\
 \hline
\multirow{3}*{Steady state $s_1$} & \multirow{3}*{53} & \multirow{3}*{3} & By direct percolation & 0.04 & 3.98 & 248.19 \\
& & & Via trap spaces & 0.06 & 13.04 & 369.84 \\
& & & Combined & 0.07 & 15.89 & 1681.83 \\
 \hline
\multirow{3}*{Steady state $s'_1$} & \multirow{3}*{53} & \multirow{3}*{3} & By direct percolation & 0.03 & 0.90 & 25.69 \\
& & & Via trap spaces & 0.07 & 2.17 & 103.20 \\
& & & Combined & 0.08 & 3.12 & 287.65 \\
 \hline
\multirow{3}*{Steady state $s_2$} & \multirow{3}*{53} & \multirow{3}*{2} & By direct percolation & 0.04 & 0.32 & 7.20 \\
& & & Via trap spaces & 0.08 & 2.46 & 155.22 \\
& & & Combined & 0.07 & 2.64 & 298.30 \\
 \hline
\multirow{3}*{Steady state $s'_2$} & \multirow{3}*{53} & \multirow{3}*{2} & Direct percolation & 0.04 & 0.93 & 6.61 \\
& & & Via trap spaces & 0.07 & 6.04 & 18.54 \\
& & & Combined & 0.08 & 6.52 & 26.59 \\
\hline
\end{tabular}
}
\label{tab:times}
\end{table}


\section{Discussion}

The method presented in this work provides a new tool for control strategy identification, based on value percolation, that uses trap spaces to identify potentially smaller control strategies that could be missed by usual percolation-based methods. This approach implements the standard node interventions acting on specific components, as well as edge interventions acting on interactions between them. Considering edge interventions widens the range of possible control strategies, for example when restrictions on the components that can be subject to intervention prevents the applicability of node control for a desired target. It can also broaden the possibilities for potential applications, for instance by allowing to act on the specific interaction between two proteins, while preserving their role in other potentially critical cell processes. The examples of edge control strategies shown in the MAPK case study illustrate the diverse and new possibilities offered by edge control (\Cref{Application}).

The formulation of these control problems as Boolean constraint problems in Answer Set Programming (ASP), extending the works from \cite{control_asp} and \cite{control_asp_trap_spaces}, aims to address the challenge of their associated combinatorial explosion. While our implementation can handle state-of-the-art biological models (\Cref{Application}), further experiments are required to fully evaluate the scalability of our extended implementation. Although in biological interaction networks the number of regulators for each component is often small in comparison to the overall number of species, which significantly helps in limiting the computational load, topological properties of the network can have a substantial impact. In particular, our approach requires the identification of some trap spaces (selected trap spaces, \Cref{main_cs}), which are used as inputs for the ASP program. The number of these selected trap spaces can be relatively high, for instance in the case of networks with many steady states, and significantly impact the running times. On the other hand, constraint programs for different selected trap spaces could be solved in parallel, with a post-processing step to ensure minimality of the results. An alternative approach could aim at identifying the relevant selected trap spaces by extending the constraint problem, rather than as a preliminary step, avoiding the costly explicit enumeration.

In this work we deal with node and edge interventions both separately and combined. When mixing the two types of interventions, it is necessary to define a priority order to avoid the inconsistency problems that might arise. Here we consider that node control takes priority over edge control and we forbid contradictory interventions targeting the same component. Further works could include different prioritisation orders and study how these might affect the controllability of a system. Another aspect that needs careful consideration is the definition of optimality for control strategies that can include both node and edge interventions. Here we considered minimality with respect to inclusion exclusively. Specific evaluations of the costs required to implement different control interventions could lead to the formulation of optimization functions more fitting to the specific model.

The use of selected trap spaces could be easily extended to other types of control. A control strategy that drives the dynamics to a selected trap space can be seen as a \emph{transient} control strategy, meaning that the intervention could be applied for a certain period of time, until the trap space is reached, and then be released. We think that trap spaces could be further exploited for the identification of more sophisticated control approaches, like sequential interventions. Given the flexibility and efficacy shown by constraint-based approaches, the extension to these problems, in particular in ASP, should be explored.

\bibliographystyle{abbrv}
\bibliography{references}

\begin{thebibliography}{10}

\bibitem{control_bcn}
C.~{Biane} and F.~{Delaplace}.
\newblock Causal reasoning on {B}oolean control networks based on abduction:
  Theory and application to cancer drug discovery.
\newblock {\em IEEE/ACM Transactions on Computational Biology and
  Bioinformatics}, 16(5):1574--1585, 2019.

\bibitem{GINsim}
C.~Chaouiya, A.~Naldi, and D.~Thieffry.
\newblock {\em Logical Modelling of Gene Regulatory Networks with {GIN}sim.},
  volume 804, pages 463--79.
\newblock 2012.

\bibitem{control_trap_spaces}
L.~Cifuentes~Fontanals, E.~Tonello, and H.~Siebert.
\newblock Control strategy identification via trap spaces in {B}oolean
  networks.
\newblock In A.~Abate, T.~Petrov, and V.~Wolf, editors, {\em Computational
  Methods in Systems Biology}, pages 159--175, Cham, 2020. Springer
  International Publishing.

\bibitem{control_asp_trap_spaces}
{Cifuentes Fontanals, L. and Tonello, E. and Siebert, H.}
\newblock {Computing Trap Space-based Control Strategies for Boolean Networks
  using Answer Set Programming}.
\newblock {\em Accepted at Proceedings of the International Conference of
  Computational Methods in Sciences and Engineering 2021 (ICCMSE-20201). To
  appear. Preprint available at
  https://github.com/Lauracf/trap-space-control/blob/master/Papers/control$\_$via$\_$trap$\_$spaces$\_$using$\_$asp.pdf},
  2021.

\bibitem{control_model_checking}
{Cifuentes Fontanals, L. and Tonello, E. and Siebert, H.}
\newblock {Control in Boolean networks with model checking}.
\newblock {\em Accepted at Frontiers in Applied Mathematics and Statistics. To
  appear. Preprint available at https://arxiv.org/abs/2112.10477}, 2021.

\bibitem{sinergies}
{\AA}.~Flobak, A.~Baudot, E.~Remy, L.~Thommesen, D.~Thieffry, M.~Kuiper, and
  A.~L{\ae}greid.
\newblock Discovery of drug synergies in gastric cancer cells predicted by
  logical modeling.
\newblock {\em PLOS Computational Biology}, 11(8):1--20, 2015.

\bibitem{clingo}
M.~Gebser, B.~Kaufmann, R.~Kaminski, M.~Ostrowski, T.~Schaub, and M.~Schneider.
\newblock Potassco: The potsdam answer set solving collection.
\newblock {\em AI Commun.}, 24(2):107–124, 2011.

\bibitem{mapk_network}
L.~Grieco, L.~Calzone, I.~Bernard-Pierrot, F.~Radvanyi, B.~Kahn-Perlès, and
  D.~Thieffry.
\newblock Integrative modelling of the influence of {MAPK} network on cancer
  cell fate decision.
\newblock {\em PLOS Computational Biology}, 9(10):1--15, 10 2013.

\bibitem{control_asp}
R.~Kaminski, T.~Schaub, A.~Siegel, and S.~Videla.
\newblock Minimal intervention strategies in logical signaling networks with
  {ASP}.
\newblock {\em Theory and Practice of Logic Programming}, 13(4-5):675--690,
  2013.

\bibitem{klarner_trap_spaces}
H.~Klarner, A.~Bockmayr, and H.~Siebert.
\newblock Computing maximal and minimal trap spaces of {B}oolean networks.
\newblock {\em Natural Computing}, 14:535--544, 2015.

\bibitem{klarner_attractor_approx}
H.~Klarner and H.~Siebert.
\newblock Approximating attractors of {B}oolean networks by iterative {CTL}
  model checking.
\newblock {\em Frontiers in Bioengineering and Biotechnology}, 3:130, 2015.

\bibitem{PyBoolNet}
H.~Klarner, A.~Streck, and H.~Siebert.
\newblock {{P}y{B}ool{N}et: a Python package for the generation, analysis and
  visualization of {B}oolean networks}.
\newblock {\em Bioinformatics}, 33(5):770--772, 2016.

\bibitem{control_basins_seq}
H.~Mandon, C.~Su, S.~Haar, J.~Pang, and L.~Paulev{\'e}.
\newblock Sequential reprogramming of {B}oolean networks made practical.
\newblock In L.~Bortolussi and G.~Sanguinetti, editors, {\em Computational
  Methods in Systems Biology}, volume 11773, pages 3--19, Cham, 2019. Springer
  International Publishing.

\bibitem{control_algebra}
D.~Murrugarra, A.~Veliz-Cuba, B.~Aguilar, and R.~Laubenbacher.
\newblock Identification of control targets in {B}oolean molecular network
  models via computational algebra.
\newblock {\em BMC Systems Biology}, 10(1):94, 2016.

\bibitem{control_intervention_sets}
R.~Samaga, A.~V. Kamp, and S.~Klamt.
\newblock Computing combinatorial intervention strategies and failure modes in
  signaling networks.
\newblock {\em Journal of Computational Biology}, 17(1):39--53, 2010.

\bibitem{control_multivalued_algebra}
L.~Sordo~Vieira, R.~Laubenbacher, and D.~Murrugarra.
\newblock Control of intracellular molecular networks using algebraic methods.
\newblock {\em Bulletin of Mathematical Biology}, 82(2), 2020.

\bibitem{cabean}
C.~Su and J.~Pang.
\newblock {{CABEAN}: a software for the control of asynchronous {B}oolean
  networks}.
\newblock {\em Bioinformatics}, 37(6):879--881, 2020.

\bibitem{Caspo}
S.~Videla, J.~Saez-Rodriguez, C.~Guziolowski, and A.~Siegel.
\newblock {Caspo: a toolbox for automated reasoning on the response of logical
  signaling networks families}.
\newblock {\em Bioinformatics}, 33(6):947--950, 2016.

\bibitem{target_control}
G.~Yang, J.~Gómez Tejeda~Zañudo, and R.~Albert.
\newblock Target control in logical models using the domain of influence of
  nodes.
\newblock {\em Frontiers in Physiology}, 9:454, 2018.

\bibitem{control_motifs}
J.~G.~T. Zañudo and R.~Albert.
\newblock Cell fate reprogramming by control of intracellular network dynamics.
\newblock {\em PLOS Computational Biology}, 11(4):1--24, 2015.

\bibitem{tlgl_network}
R.~Zhang, M.~V. Shah, J.~Yang, S.~B. Nyland, X.~Liu, J.~K. Yun, R.~Albert, and
  T.~P. Loughran.
\newblock Network model of survival signaling in large granular lymphocyte
  leukemia.
\newblock {\em Proceedings of the National Academy of Sciences},
  105(42):16308--16313, 2008.

\end{thebibliography}

\end{document}